\documentclass[onecolumn,12pt,draftcls]{IEEEtran}
\IEEEoverridecommandlockouts

\usepackage{amsmath}
\usepackage{amssymb}
\usepackage{amsthm}
\usepackage{dsfont}
\usepackage{enumitem}
\usepackage{bm}
\usepackage{hyperref}

\usepackage{graphicx}
\usepackage{xcolor}

\newcommand{\E}[2]{\mathbb{E} _{ #1 }  \left[ #2 \right]}
\newcommand{\Var}[2]{\mathrm{Var} _{#1} \left( #2 \right)}

\newcommand{\Prob}[2]{\mathbb{P} _{ #1 } \left\{ #2 \right\}}
\newcommand{\Ep}[1]{\mathrm{E}_{p}  \left[ #1 \right]}
\newcommand{\Vp}[1]{\mathrm{V}_{p} \left( #1 \right)}

\newcommand{\floor}[1]{\left\lfloor #1 \right\rfloor}
\newcommand{\norm}[1]{\left\lVert #1 \right\rVert}
\newcommand{\abs}[1]{\left| #1 \right|}
\newcommand{\R}{\mathbb{R}}

\newcommand{\td}{\Tilde}
\newcommand{\KL}[2]{D(#1||#2)}

\newcommand{\brc}[1]{\left( #1 \right)}
\newcommand{\eps}{\epsilon}

\newcommand{\MISE}{\text{MISE}}

\newcommand{\FAR}[1]{\mathrm{FAR} \left( #1 \right)}
\newcommand{\WADD}[1]{\mathrm{WADD} \left( #1 \right)}

\newcommand{\WADDth}[1]{\mathrm{WADD} _{ \theta } \left( #1 \right)}

\DeclareMathOperator*{\esssup}{ess\,sup}

\newtheorem{theorem}{Theorem}[section]

\newtheorem{lemma}[theorem]{Lemma}

\theoremstyle{definition}

\theoremstyle{remark}
\newtheorem*{remark}{Remark}    

\begin{document}

\title{Quickest Change Detection with Leave-one-out Density Estimation}

\author{Yuchen Liang, ~\IEEEmembership{Student Member,~IEEE}, and  Venugopal V. Veeravalli, ~\IEEEmembership{Fellow, ~IEEE}

\thanks{The authors are with the ECE Department
of the University of Illinois at Urbana-Champaign. Email: \{yliang35, vvv\}@ILLINOIS.EDU. 
}
}
\maketitle

\vspace*{-0.5in}

\begin{abstract}
The problem of quickest change detection in a sequence of independent observations is considered. The pre-change distribution is assumed to be known, while the post-change distribution is completely unknown. A window-limited leave-one-out (LOO) CuSum test is developed, which does not assume any knowledge of the post-change distribution, and does not require any post-change training samples. It is shown that, with certain convergence conditions on the density estimator, the LOO-CuSum test is first-order asymptotically optimal, as the false alarm rate goes to zero. The analysis is validated through numerical results, where the LOO-CuSum test is compared with baseline tests that have distributional knowledge.
\end{abstract}


\section{Introduction}

The problem of quickest change detection (QCD) is of fundamental importance in mathematical statistics (see, e.g., \cite{vvv_qcd_overview,xie_vvv_qcd_overview} for an overview). Given a sequence of observations whose distribution changes at some unknown change-point, the goal is to detect the change in distribution as quickly as possible after it occurs, while controlling the false alarm rate. 
In classical formulations of the QCD problem, it is assumed that the pre- and post-change distributions are known, and that the observations are independent and identically distributed (i.i.d.) in either the pre-change or the post-change regime. However, in many practical situations, while it is reasonable to assume that we can accurately estimate the pre-change distribution, the post-change distribution is rarely completely known.


There have been extensive efforts to address pre- and/or post-change distributional uncertainty in QCD problems. In the case where both distributions are not fully known, one approach is to assume that they are indexed by a (low-dimensional) parameter that comes from a pre-defined parameter set, and employ a generalized likelihood ratio (GLR) approach to detection -- this was first introduced in \cite{lorden1971} and later analyzed in more detail in \cite{lai1998}.
In particular, in \cite{lai1998}, it is assumed that the pre-change distribution is known and that the post-change distribution comes from a parametric family, with the parameter being finite-dimensional. A window-limited GLR test is proposed, which is shown to be asymptotically optimal under certain smoothness conditions.
This work has recently been extended to non-stationary post-change settings \cite{non_stat_glr_2021}.


Another approach to dealing with distributional uncertainty in QCD problems is the minimax robust approach \cite{huber1965}, where it is assumed that the pre- and post-change distributions come from (known) mutually exclusive uncertainty classes, and the goal is to optimize the performance for the worst-case choice of distributions in the uncertainty classes. Under certain conditions, e.g., joint stochastic boundedness (see, e.g., \cite{moulin-veeravalli-2018} for a definition) and weak stochastic boundedness \cite{Molloy2017}, robust solutions can be found \cite{Unnikrishnan2011, Molloy2017}. However, these robust tests can have suboptimal performance for the actual distributions encountered in practice.

In this paper, we will assume complete knowledge of the pre-change distribution, while not making any parametric assumptions about the post-change distribution. 
There have also been approaches to deal with non-parametric uncertainty in the distributions in QCD problems. One approach is to replace the log-likelihood ratio by some other statistic and formulate the test in the non-parametric setting. Examples of this approach include the use of kernel M-statistics \cite{xie_mstat, kernelcusum}, one-class SVMs \cite{one_svm}, nearest neighbors \cite{chu2018sequential, nn_2019}, and Geometric Entropy Minimization \cite{yilmaz_2017, kurt_2020}. In \cite{xie_mstat}, a test is proposed that compares the kernel maximum mean discrepancy (MMD) within a window to a given threshold. A way to set the threshold is also proposed that meets the false alarm rate asymptotically \cite{xie_mstat}. Another approach is to estimate the log-likelihood ratio and thus the CuSum test statistic through a pre-collected training set. The include direct kernel estimation \cite{sugiyama} and, more recently, neural network estimation \cite{moustakides2019}. 
However, \textit{the tests proposed in \cite{xie_mstat}--\cite{moustakides2019} lack explicit performance guarantees on the detection delay.}
The closest work to ours is \cite{binning}, where a binning approach is proposed to solve the QCD problem asymptotically without any pre-collected training set. The asymptotic optimality is established for the case where the pre-change distribution is known, the post-change distribution is distinguishable from the pre-change with binning, and both distributions have discrete support \cite{binning}.

Our contributions are as follows:
\begin{enumerate}
    \item We propose a window-limited leave-one-out (LOO) CuSum test, which does not assume any knowledge of the post-change distribution, and does not require any post-change training samples.
    \item We provide a way to set the test threshold that asymptotically meets the false alarm constraint.
    \item We show that the proposed LOO-CuSum test is first-order asymptotically optimum, as the false alarm rate goes to zero.
    \item We validate our analysis through numerical results, in which we compare the LOO-CuSum test with baseline tests that have distributional knowledge.
\end{enumerate}

The rest of the paper is structured as follows. In Section~\ref{sec:de_property}, we describe several properties required of the density estimators for asymptotically optimal detection. In Section~\ref{sec:qcd_loo}, we propose the LOO-CuSum test, and analyze its theoretical performance.
In Section~\ref{sec:num_res}, we present numerical results that validate the theoretical analysis.
Finally, we conclude the paper in Section~\ref{sec:concl}.

\section{Leave-one-out (LOO) Density Estimator}
\label{sec:de_property}

Let $X_1,X_2,\dots \in \R^d$ be i.i.d. samples drawn from an unknown distribution $p$. Denote by $\text{supp}(p)$ the support of $p$. Denote by $\mathrm{E}_{p}$ and $\mathrm{V}_{p}$, the expectation and variance operator, respectively, under $p$.
Denote the estimated density as $\widehat{p}^{n,k}_{-i}$, where the subscript $-i$ represents that $X_i$,  with $k \leq i \leq n$, is the sample that is left out.
Note that the density is a function of $X^{n,k}_{-i} := X_k,\dots X_{i-1},X_{i+1},\dots,X_n$, and thus $\widehat{p}^{n,k}_{-i}$ and $X_i$ are independent. The estimation procedure is assumed to be sample-homogeneous, i.e., $\widehat{p}^{n,k}_{-i} \stackrel{d.}{=} \widehat{p}^{n,k}_{-j}, \forall k \leq i < j \leq n$. 

The Kullback-Leibler (KL) divergence between distributions $p$ and $q$ is
\[ \KL{p}{q} := \int_{\text{supp}(p)} \log (p(x)/q(x)) d x. \]

Suppose that, for large enough $n$, there exists constants $0 < \beta_1,C_1,C_2 < \infty$ and $1 < \beta_2 < 2$ (that depend only on the distribution $p$ and the estimation procedure) such that, for each $1 \leq i \leq n+1$, the KL loss \cite{de-with-kl-loss} of the leave-one-out (LOO) estimator satisfies
\begin{equation}
\label{eq:converg_m1_upper_bound}
    \text{KL-loss} (\widehat{p}^{n,k}_{-i}) := \Ep{\KL{p}{\widehat{p}^{n,k}_{-i}}} \leq \frac{C_1}{(n-k)^{\beta_1}}
\end{equation}
where the expectation $\mathrm{E}_p$ is over the randomness of $\widehat{p}^{n,k}_{-i}$.
Also, the total variance of the estimator satisfies
\begin{equation}
\label{eq:converg_m2_bound}
    \Vp{\sum_{i=k}^{n} \log \frac{p(X_i)}{\widehat{p}^{n,k}_{-i}(X_i)}} \leq C_2 (n-k+1)^{\beta_2}.
\end{equation}

One typical loss measure for a density estimator is the mean-integrated squared error (MISE), defined as (see, e.g., \cite[Chap.~2]{mult-denst-est})
\begin{align}
\label{eq:mise_def}
   \MISE(p, \widehat{p}^{n,k}_{-i}) &= \Ep {\int (\widehat{p}^{n,k}_{-i} (x_i) - p(x_i))^2 d x_i} 
   = \Ep {\norm{\widehat{p}^{n,k}_{-i} - p}_2^2},\ \forall k \leq i \leq n.
\end{align}

The following lemma connects the MISE with the bounds in \eqref{eq:converg_m1_upper_bound} and \eqref{eq:converg_m2_bound}.

\begin{lemma}
\label{lem:Dest}
Suppose that there exist $\overline{\zeta}, \underline{\zeta}$ such that 
\begin{equation} \label{eq:compact_support}
    0 < \underline{\zeta} \leq p(x), \widehat{p}^{n,k}_{-i}(x) \leq \overline{\zeta} < \infty,~\forall x \in \text{supp}(p).
\end{equation}
for any $k \leq i \leq n$.
If the estimator achieves
\begin{equation}
    \MISE(p, \widehat{p}^{n,k}_{-i}) \leq \frac{C_3}{(n-k)^{\beta_1}}
\end{equation}
for some constant $C_3 < \infty$, then \eqref{eq:converg_m1_upper_bound} and \eqref{eq:converg_m2_bound} are satisfied with
\[ C_1 = \frac{C_3}{\underline{\zeta}},\quad C_2 = \frac{\overline{\zeta} C_3}{\underline{\zeta}^2}, \quad \beta_2 = 2 - \beta_1. \]
\end{lemma}

\begin{proof}
Write $\widehat{p}_i(X_i) = \widehat{p}^{n,k}_{-i}(X_i)$. For the first proof, we use the fact that
$\log{s} \leq s-1$.
Thus,
\begin{align*}
    \Ep{\log\frac{p(X_i)}{\widehat{p}_i(X_i)}} &\leq \Ep{\frac{p(X_i)}{\widehat{p}_i(X_i)} - 1}\\
    &= \Ep{\int \frac{p^2(x_i) - p(x_i) \widehat{p}_i(x_i)}{\widehat{p}_i(x_i)} d x_i}\\
    &\stackrel{(*)}{=} \Ep{\int \frac{p^2(x_i) - 2 p(x_i) \widehat{p}_i(x_i) + \widehat{p}_i^2(x_i)}{\widehat{p}_i(x_i)} d x_i}\\
    &\leq \frac{1}{\underline{\zeta}} \MISE(p, \widehat{p}^{n,k}_{-i})
\end{align*}
which shows \eqref{eq:converg_m1_upper_bound}, where $(*)$ follows because both $p$ and $\widehat{p}_i$ are densities given $x_{-i}$.
For the variance, first note that $(\log{s})^2 \leq r (s-1)^2$ on $s \geq \underline{\zeta} / \overline{\zeta}$ if $r := (\ln (\underline{\zeta} / \overline{\zeta}))^2 / ((\underline{\zeta} / \overline{\zeta})-1)^2$. Thus,
\begin{align*}
    \Vp{\log\frac{p(X_i)}{\widehat{p}_i(X_i)}} &\leq \Ep{r^2\left(\frac{p(X_i)}{\widehat{p}_i(X_i)} - 1\right)^2}\\
    &= r^2 \Ep{\int \frac{\left(p(x_i) - \widehat{p}_i(x_i)\right)^2}{\widehat{p}^2_i(x_i)} p(x_i) d x_i}\\
    &\leq \frac{\overline{\zeta} r^2}{\underline{\zeta}^2} \MISE (p, \widehat{p}^{n,k}_{-i})
\end{align*}
Therefore,
\begin{align*}
    \Vp{\sum_{i=k}^n \log\frac{p(X_i)}{\widehat{p}_i(X_i)}} &\leq (n-k+1)^2 \Vp{\log\frac{p(X_i)}{\widehat{p}_i(X_i)}}\\
    &\leq \frac{\overline{\zeta} r^2 C_3 }{\underline{\zeta}^2} (n-k+1)^{2-\beta_1}
\end{align*}
and thus $\beta_2 = 2-\beta_1$.
The proof is now complete. \qedhere
\end{proof}

An example of a LOO estimator that satisfies \eqref{eq:converg_m1_upper_bound} and \eqref{eq:converg_m2_bound}  (under condition \eqref{eq:compact_support}) is the LOO \emph{kernel} density estimator (LOO-KDE), defined as
\begin{equation}
\label{def:kde}
    \widehat{p}^{n,k}_{-i}(x_i) = \frac{1}{(n-k)h} \sum_{\substack{j=k \\ j \neq i}}^n K\left(\frac{x_i-x_j}{h}\right)
\end{equation}
where $K(\cdot) \geq 0$ is a kernel function and $h > 0$ is a smoothing parameter.

The KL loss for kernel density estimators is analyzed carefully in \cite{de-with-kl-loss}, where it is shown that the rate of convergence in KL loss is slower than that of MISE for most well-behaved densities. Nevertheless, this rate indeed converges to zero with a polynomial decay rate with the use of appropriate kernel functions,  and thus \eqref{eq:converg_m1_upper_bound} is satisfied. Furthermore, using \eqref{eq:compact_support}, it can be shown that \eqref{eq:converg_m2_bound} is also satisfied. We note that the actual choices of $\beta_1$ and $\beta_2$ do not affect the first-order asymptotic optimality result given in Thm~\ref{thm:opt}.

\section{QCD with LOO Density Estimation}
\label{sec:qcd_loo}

Let $X_1,X_2,\dots,X_n,\dots \in \R^d$ be a sequence of independent random variables (or vectors), and let $\nu$ be a change-point. Assume that $X_1, \dots, X_{\nu-1}$ all have density $p_0$ with respect to some measure $\mu$. Furthermore, assume that $X_\nu, X_{\nu+1}, \dots$ have densities $p_1$ also with respect to $\mu$. Here $p_0$ is assumed to be completely known. While $p_1$ is completely unknown, we assume that \eqref{eq:converg_m1_upper_bound} and \eqref{eq:converg_m2_bound} are satisfied for LOO estimators of $p_1$.

Let $\mathbb{P}_\nu$ denote the probability measure on the entire sequence of observations when the change-point is $\nu$, and let $\E{\nu}{\cdot}$ denote the corresponding expectation.
The change-time $\nu$ is assumed to be unknown but deterministic. The problem is to detect the change quickly, while controlling the false alarm rate. Let $\tau$ be a stopping time \cite{moulin-veeravalli-2018} defined on the observation sequence associated with the detection rule, i.e. $\tau$ is the time at which we stop taking observations and declare that the change has occurred.

\subsection{Classical Results}

When $p_1$ is known, Lorden \cite{lorden1971} proposed solving the following optimization problem to find the best stopping time $\tau$:
\begin{equation}
\label{prob_def}
    \inf_{\tau \in \mathcal{C}_\alpha} \WADD{\tau}
\end{equation}
where
\begin{equation}
    \WADD{\tau} := \sup_{\nu \geq 1} \esssup \E{\nu}{\left(\tau-\nu+1\right)^+|{\cal F}_{\nu-1}}
\end{equation}
characterizes the worst-case delay, and ${\cal F}_n$ denotes the sigma algebra generated by $X_1,\dots,X_n$, i.e., ${\cal F}_n = \sigma(X_1,\dots,X_n)$. The constraint set in  \eqref{prob_def} is
\begin{equation}
\label{fa_constraint}
    \mathcal{C}_\alpha := \left\{ \tau: \FAR{\tau} \leq \alpha \right\}
\end{equation}
with $\FAR{\tau} := \frac{1}{ \E{\infty}{\tau}}$
which guarantees that the false alarm rate of the algorithm does not exceed $\alpha$. Here, $\E{\infty}{\cdot}$ is the expectation operator when the change never happens, and $(\cdot)^+:=\max\{0,\cdot\}$.

Lorden also showed that Page's Cumulative Sum (CuSum) algorithm \cite{page1954} whose test statistic is given by:
\begin{equation*}
    W(n) = \max_{1\leq k \leq n+1} \sum_{i=k}^n \log \frac{p_1(X_i)}{p_0(X_i)} 
\end{equation*}
solves the problem in \eqref{prob_def} asymptotically as $\alpha \to 0$.
The CuSum stopping rule is given by:
\begin{equation}
\label{def:cusum}
    \tau_{\text{Page}}\left(b\right) := \inf \{n:W(n)\geq b \}
\end{equation}
where the threshold is set as $b = \abs{\log \alpha}$. It was shown by Moustakides \cite{moustakides1986} that the CuSum algorithm is exactly optimal for the problem in (\ref{prob_def}).
The asymptotic performance is
\begin{equation}
    \inf_{\tau \in \mathcal{C}_\alpha} \WADD{\tau} \sim \WADD{\tau_{\text{Page}}\left(\abs{\log\alpha}\right)} \sim \frac{\abs{\log \alpha}}{\KL{p_1}{p_0}}
\end{equation}
as $\alpha \to 0$.
Here 
$Y_\alpha\sim G_\alpha$ is equivalent to $Y_\alpha = G_\alpha (1+o(1))$. Also, we use the notation $o(1)$ to denote a quantity that goes to $0$, as $\alpha \to 0$ or $b \to \infty$.

When the post-change distribution has parametric uncertainties, Lai \cite{lai1998} generalized this performance guarantee with the following assumptions. Suppose that $p_0$ and $p_1$ satisfy
\begin{equation}
\label{eq:lai_upper}
    \sup_{\nu \geq 1} \Prob{\nu}{\max_{t \leq n} \sum_{i=\nu}^{\nu+t} Z_i \geq (1+\delta) n I} \xrightarrow{n \to \infty} 0
\end{equation}
for any $\delta > 0$, and
\begin{equation}
\label{eq:lai_lower}
    \sup_{t \geq \nu \geq 1} \Prob{\nu}{\sum_{i=t}^{t+n} Z_i \leq (1-\delta) n I} \xrightarrow{n \to \infty} 0
\end{equation}
for any $\delta \in (0,1)$, with some constant $I > 0$. Also, suppose that the window size $m_\alpha$ satisfies
\begin{equation*}
\label{eq:lai_malpha}
    \liminf m_\alpha / \abs{\log\alpha} > I^{-1}\quad \text{ and } \log m_\alpha = o(\abs{\log\alpha}).
\end{equation*}
Then, the window-limited GLR CuSum test:
\begin{equation}
    \td{\tau}_{GLR}\left(b\right) := \inf \left\{n:\max_{n-m_\alpha \leq k \leq {n+1}} \sup_{\theta \in \Theta} \sum_{i=k}^n Z^{\theta}_{i,k} \geq b \right\}
\end{equation}
with test threshold $b_\alpha = \abs{\log\alpha}$ solves the problem in \eqref{prob_def} asymptotically as $\alpha \to 0$. The asymptotic performance is
\begin{equation}
\label{eq:lai_perf}
    \inf_{\tau \in \mathcal{C}_\alpha} \WADDth{\tau} \sim \WADDth{\td{\tau}_{\text{Page}} \left(b_\alpha\right)} \sim \frac{\abs{\log \alpha}}{I}.
\end{equation}
Note that $I = \KL{p_1}{p_0}$ when $p_0$ and $p_1$ are independent.

\subsection{Leave-one-out CuSum Test}

For the case when $p_1$ is unknown, we define the leave-one-out log-likelihood ratio as
\begin{equation}
\label{eq:est_llr}
    \widehat{Z}^{n,k}_i = \log \frac{\widehat{p}^{n,k}_{-i}(X_i)}{p_0(X_i)},\ \forall k \leq i \leq n.
\end{equation}
and the leave-one-out (LOO) CuSum stopping rule as
\begin{equation}
\label{def:loo_cusum}
    \widehat{\tau}(b) := \inf \left\{n:\max_{(n-m_\alpha)^+ \leq k \leq n-1} \sum_{i=k}^n \widehat{Z}^{n,k}_i \geq b \right\}.
\end{equation}
Here the window size $m_\alpha$ is designed to satisfy
\begin{equation}
\label{eq:malpha}
    \liminf m_\alpha / \abs{\log\alpha} > f I^{-1}~\text{with}~ \log m_\alpha = o(\abs{\log\alpha})
\end{equation}
where $f > 1$ is a constant.

In Lemma~\ref{lem:fa}, we show that $\widehat{\tau}$ with a properly chosen threshold $b_\alpha$ satisfies the false alarm constraint in \eqref{fa_constraint} asymptotically. In Lemma~\ref{lem:delay}, we establish an asymptotic upper bound on $\WADD{\widehat{\tau}(b_\alpha)}$. Finally, in Theorem~\ref{thm:opt}, we combine the two lemmas and establish the first-order asymptotic optimality of $\widehat{\tau}(b_\alpha)$. 


\begin{lemma}
\label{lem:fa}
Suppose that $b_\alpha$ satisfies
\begin{equation}
\label{eq:nonparam_b_alpha}
    b_\alpha = \abs{\log\alpha} + \log(8 m_\alpha).
\end{equation}
Then,
\begin{equation*}
    \E{\infty}{\widehat{\tau}(b_\alpha)} \geq \alpha^{-1}.
\end{equation*}
\end{lemma}
\begin{remark}
If $m_\alpha$ satisfies \eqref{eq:malpha}, then $b_\alpha = \abs{\log\alpha} (1+o(1))$ as $\alpha \to 0$.
\end{remark}

\begin{proof}
Fix $\nu \geq 1$. For all threshold $b > 0$,
\begin{align}
    \Prob{\infty}{\nu \leq \widehat{\tau}(b) < \nu + m_\alpha } &\leq \sum_{k = \nu-m_\alpha}^{\nu+m_\alpha-1} \Prob{\infty}{\sum_{i=k}^n \widehat{Z}^{n,k}_i \geq b,\ \exists n: k \leq n \leq k + m_\alpha } \nonumber\\
    &=: \sum_{k=\nu-m_\alpha}^{\nu+m_\alpha-1} \Prob{\infty}{\tau_k(b) \leq k + m_\alpha}
\end{align}
where we define an auxiliary stopping time $\tau_k(b)$ as
\begin{equation*}
    \tau_k(b) := \inf\left\{n \geq k: \sum_{i=k}^n \widehat{Z}^{n,k}_i \geq b \right\}.
\end{equation*}
Now, for any $b > 0$,
\begin{align}
    \Prob{\infty}{\tau_k(b) \leq k + m_\alpha} &= \int_{\{\tau_k(b) \leq k + m_\alpha\}} d \mathbb{P}_\infty \nonumber\\
    &= \int_{\{\tau_k(b) \leq k + m_\alpha\}} \prod_{i=k}^{\tau_k(b)} \frac{\widehat{p}^{n,k}_{-i}(x_i)}{p_0(x_i)} \prod_{i=k}^{\tau_k(b)} \frac{p_0(x_i)}{\widehat{p}^{n,k}_{-i}(x_i)} d \mathbb{P}_\infty \nonumber\\
    &\leq e^{-b} \E{\infty}{\sum_{i=k}^{k+m_\alpha} e^{\widehat{Z}^{n,k}_i}(X_i)} \nonumber\\
    &= e^{-b} m_\alpha
\end{align}
where the second to last line follows by the definition of $\tau_k(b)$, and the last equality follows because
\begin{equation*}
    \E{\infty}{\frac{\widehat{p}^{n,k}_{-i}(X_i)}{p_0(X_i)}}\\
    = \int \underbrace{\brc{\int \frac{\widehat{p}^{n,k}_{-i}(x_i)}{p_0(x_i)} p_0(x_i) d x_i}}_{= 1} p_0(x_{-i}) d x_{-i} = 1.
\end{equation*}
Therefore,
\[ \sup_{\nu \geq 1} \Prob{\infty}{\nu \leq \widehat{\tau}(b) < \nu + m_\alpha } \leq 2 m_\alpha^2 e^{-b}, \]
and by \cite[Lemma~2.2(ii)]{tartakovsky_qcd2020},
\[ \E{\infty}{\widehat{\tau}(b)} \geq \frac{1}{8 m_\alpha} e^b \]
for all $b > 0$. Choosing $b = b_\alpha$ gives the desired result. \qedhere
\end{proof}

\begin{lemma}
\label{lem:delay}
Let $b_\alpha = \abs{\log\alpha}(1+o(1))$ and $m_\alpha$ satisfy \eqref{eq:malpha}. Suppose that \eqref{eq:converg_m1_upper_bound} and \eqref{eq:converg_m2_bound} hold.
Further, suppose \eqref{eq:lai_lower} holds. Then,
\begin{equation*}
    \WADD{\widehat{\tau}(b_\alpha)} \leq \frac{\abs{\log{\alpha}}}{\KL{p_1}{p_0}} (1+o(1))
\end{equation*}
as $\alpha \to 0$.
\end{lemma}

\begin{proof}
Write $I = \KL{p_1}{p_0}$. Let $\delta_0 := 1-f^{-1}$ and $\delta_b \in (0,\delta_0)$ be a function of $b$ such that $\delta_b \searrow 0$ as $b \nearrow \infty$. Define
\begin{equation}
\label{eq:delay_nb}
    n_b := \floor{\frac{b}{I(1-\delta_b)}}.
\end{equation}
By definition of the window size $m_\alpha$ in \eqref{eq:malpha}, if $b_\alpha = \abs{\log\alpha}(1+o(1))$, then
\begin{equation*}
    n_b < \frac{b}{I(1-\delta_0)} = \frac{f b}{I} \leq m_\alpha
\end{equation*}
for all sufficiently small $\alpha$. Suppose for now that we can get a large enough $b$ to satisfy
\begin{equation}
\label{eq:delay_main}
    \sup_{t \geq \nu \geq 1} \esssup \Prob{\nu}{\left.\sum_{i=t}^{t+n_b-1} \widehat{Z}^{t+n_b-1,t}_i < b\right| {\cal F}_{t-1}} < 2 \delta_b^2.
\end{equation}
Then, $\forall \nu \geq 1$, $\forall k \geq 1$, and for all large $b$,
\begin{align}
    &\esssup \Prob{\nu}{\widehat{\tau}(b)-\nu+1 > k n_b | \widehat{\tau}(b)-\nu+1 > (k-1) n_b, {\cal F}_{\nu-1} } \nonumber\\
    &\stackrel{(i)}{\leq} \esssup \Prob{\nu}{\left.\widehat{\tau}(b)-\nu+1 > k n_b \right| {\cal F}_{\nu+(k-1)n_b-1} } \nonumber\\
    &\stackrel{(ii)}{\leq} \esssup \Prob{\nu}{\left.\sum_{i=\nu+(k-1)n_b}^{\nu+k n_b - 1} \widehat{Z}^{\nu+k n_b-1, \nu+(k-1)n_b}_i < b \right| {\cal F}_{\nu+(k-1)n_b-1} } \nonumber\\
    &\stackrel{(iii)}{\leq} 2 \delta_b^2
\end{align}
In the series of inequalities above, $(i)$ is by definition of essential supremum, $(ii)$ is explained below, and $(iii)$ follows from \eqref{eq:delay_main}. For $(ii)$, consider the complement of the two events. Since $n_b \leq m_\alpha$, we have
\begin{equation*}
    \nu+(k-1)n_b \geq \nu+k n_b - 1 - m_\alpha
\end{equation*}
which implies that $\nu+(k-1)n_b$ is still one of the hypothetical change-points. If
\[ \sum_{i=\nu+(k-1)n_b}^{\nu+k n_b - 1} \widehat{Z}^{\nu+k n_b-1, \nu+(k-1)n_b}_i \geq b,\]
then $\widehat{\tau}(b)-\nu+1 = k n_b$. Thus, from recursion,
\[ \esssup \Prob{\nu}{\widehat{\tau}(b)-\nu+1 > k n_b | {\cal F}_{\nu-1} } \leq (2 \delta_b^{2})^k.\]
Therefore, for sufficiently large $b$,
\begin{align*}
    &\sup_{\nu \geq 1} \esssup \E{\nu}{n_b^{-1} (\widehat{\tau}(b)-\nu+1)^+|{\cal F}_{\nu-1}} \nonumber\\
    &\leq \sum_{k=1}^\infty \Prob{\nu}{n_b^{-1} (\widehat{\tau}(b)-\nu+1)^+ > k | {\cal F}_{\nu-1}} \nonumber\\
    &\leq \sum_{k=0}^\infty (2 \delta_b^{2})^k = \frac{1}{1-2 \delta_b^{2}}
\end{align*}
and, as $b \to \infty$,
\begin{equation*}
    \WADD{\widehat{\tau}(b)} \leq \frac{n_b}{1-2 \delta_b^2} \leq \frac{b}{I (1-\delta_b) (1-2 \delta_b^2)} = \frac{b}{I} (1+o(1)).
\end{equation*}

It remains to show \eqref{eq:delay_main}. Note that ${\cal F}_{t-1}$ can be dropped by independence between ${\cal F}_{t-1}$ and $\widehat{Z}^{t+n_b-1,t}_i$. For any $t \geq \nu \geq 1$ and $\eps > 0$,
\begin{align}
\label{eq:delay_main2}
    &\Prob{\nu}{\sum_{i=t}^{t+n_b-1} \widehat{Z}^{t+n_b-1,t}_i < b} \nonumber\\
    &= \Prob{\nu}{\sum_{i=t}^{t+n_b-1} \widehat{Z}^{t+n_b-1,t}_i < b, \sum_{i=t}^{t+n_b-1} Z_i - \widehat{Z}^{t+n_b-1,t}_i \leq \eps} + \nonumber\\
    &\qquad \Prob{\nu}{\sum_{i=t}^{t+n_b-1} \widehat{Z}^{t+n_b-1,t}_i < b, \sum_{i=t}^{t+n_b-1} Z_i - \widehat{Z}^{t+n_b-1,t}_i \geq \eps} \nonumber\\
    &\leq \Prob{\nu}{\sum_{i=t}^{t+n_b-1} Z_i \leq b + \eps} + \Prob{\nu}{\frac{1}{n_b} \sum_{i=t}^{t+n_b-1} \left(Z_i - \widehat{Z}^{t+n_b-1,t}_i\right) \geq \frac{\eps}{n_b}} \nonumber\\
    &= \Prob{1}{\sum_{i=1}^{n_b} Z_i \leq b + \eps} + \Prob{1}{\frac{1}{n_b} \sum_{i=1}^{n_b} \left(Z_i - \widehat{Z}^{n_b,1}_i\right) \geq \frac{\eps}{n_b}}
\end{align}
Now, the first term grows larger with large $\eps$, while the second term grows smaller. It is important to choose a proper $\eps = \eps_b$ in order to keep both terms small. The idea in the following is that, we choose $\eps_b$ by controlling the second term, and then verify that it is small enough also for the first for large $b$.

For the second term in \eqref{eq:delay_main2}, write $\widehat{p}_i$ and $\widehat{Z}_i$ as short-hand notations for $\widehat{p}^{n_b,1}_{-i}$ and $\widehat{Z}^{n_b,1}_{-i}$, respectively. Note that $\E{1}{Z_i - \widehat{Z}_i} = \E{1}{\KL{p_1}{\widehat{p}_i}}$. Under the conditions for the estimator in \eqref{eq:converg_m1_upper_bound} and \eqref{eq:converg_m2_bound}, the mean and variance of $n_b^{-1} \sum_{i=1}^{n_b} (Z_i - \widehat{Z}_i)$ can be bounded as
\begin{align*}
    \E{1}{\frac{1}{n_b} \sum_{i=1}^{n_b} \log\frac{p_1(X_i)}{\widehat{p}_i(X_i)}} = \E{1}{\KL{p_1}{\widehat{p}_i}} \leq \frac{C_1}{n_b^{\beta_1}} \\
    \Var{1}{\frac{1}{n_b} \sum_{i=1}^{n_b} \log\frac{p_1(X_i)}{\widehat{p}_i(X_i)}} \leq \frac{C_2}{n_b^{2-\beta_2}}.
\end{align*}
Using Chebyhsev's inequality,
\begin{align}
\label{eq:delay_term2}
    &\Prob{1}{\frac{1}{n_b} \sum_{i=1}^{n_b} \left(Z_i - \widehat{Z}_i\right) \geq \frac{\eps}{n_b}} \nonumber\\
    &\leq \Prob{1}{\abs{\frac{1}{n_b} \sum_{i=1}^{n_b} \left(Z_i - \widehat{Z}_i\right) - \E{1}{\KL{p_1}{\widehat{p}_i}}} \geq \frac{\eps}{n_b} - \E{1}{\KL{p_1}{\widehat{p}_i}}} \nonumber\\
    &\leq \Var{1}{\frac{1}{n_b} \sum_{i=1}^{n_b} \log\frac{p_1(X_i)}{\widehat{p}_i(X_i)}} \left(\frac{\eps}{n_b} - \E{1}{\KL{p_1}{\widehat{p}_i}}\right)^{-2} \nonumber\\
    &\leq \frac{C_2}{n_b^{2-\beta_2}} \left(\frac{\eps}{n_b} - \E{1}{\KL{p_1}{\widehat{p}_i}}\right)^{-2}
\end{align}
for any $\eps > n_b \times \E{1}{\KL{p_1}{\widehat{p}_i}}$.
Now, \eqref{eq:delay_term2} is less than or equal to $\delta_b^2$ if we let
\begin{equation}
\label{eq:delay_eps_eqn}
    \frac{\eps}{n_b} - \E{1}{\KL{p_1}{\widehat{p}_i}} \geq \frac{\sqrt{C_2}}{\delta_b n_b^{\frac{2-\beta_2}{2}}} \iff \eps \geq \frac{\sqrt{C_2} n_b^{\beta_2 / 2}}{\delta_b} + n_b \E{1}{\KL{p_1}{\widehat{p}_i}}.
\end{equation}
Among the two terms above, the second term $\leq C_1 n_b^{1-\beta_1}$, so there are two cases depending on the rate for the first term.
\begin{itemize}
\item Case 1: $4 \beta_1 + \beta_2 > 2$. We want the first term in \eqref{eq:delay_eps_eqn} to be dominant by choosing a proper $\delta_b$. Let
\begin{equation}
\label{eq:delay_choose_eps}
    \eps_b = \frac{2 \sqrt{C_2} n_b^{\beta_2 / 2}}{\delta_b}.
\end{equation}
Then, the first term in \eqref{eq:delay_main2} becomes
\begin{align*}
    & \Prob{1}{\sum_{i=1}^{n_b} Z_i < b + \eps_b} \nonumber\\
    &= \Prob{1}{\sum_{i=1}^{n_b} Z_i < (1-\delta_b) n_b I + \frac{2 \sqrt{C_2} n_b^{\beta_2 / 2}}{\delta_b}}\\
    &= \Prob{1}{\sum_{i=1}^{n_b} Z_i < (1-\delta_b) n_b I \brc{1 + \frac{2 \sqrt{C_2} }{(1-\delta_b)\delta_b n_b^{(2-\beta_2) / 2} I} } } \nonumber\\
    &\leq \Prob{1}{\sum_{i=1}^{n_b} Z_i < (1-\delta_b) n_b I \brc{1 + \frac{2 f \sqrt{C_2}}{\delta_b n_b^{(2-\beta_2) / 2} I } } }
\end{align*}
where the last inequality uses $1-\delta_b > f^{-1}$. Let
\begin{equation}
\label{eq:delay_choose_delta}
    \delta_b = \frac{(4 f^2 C_2)^\frac{1}{4}}{n_b^{(2-\beta_2) / 4} \sqrt{I}}
\end{equation}
which, by \eqref{eq:delay_choose_eps}, also implies that $\eps_b \sim n_b^{\frac{2 + \beta_2}{4}}$.
As $b \to \infty$, since $\beta_2 < 2$ and $n_b \to \infty$, $\delta_b \searrow 0$ and $\eps_b$ increases slower than the linear rate. With this $\delta_b$,
\[ 1 + \frac{2 f \sqrt{C_2}}{\delta_b n_b^{(2-\beta_2) / 2} I } = 1 + \delta_b \]
and
\begin{equation}
\label{eq:delay_true_llr}
    \Prob{1}{\sum_{i=1}^{n_b} Z_i < b + \eps_b} \leq \Prob{1}{\sum_{i=1}^{n_b} Z_i < (1-\delta_b^2) n_b I}.
\end{equation}

We want to verify that the first term in \eqref{eq:delay_eps_eqn} is indeed dominant with the $\delta_b$ defined in \eqref{eq:delay_choose_delta}. Since its rate satisfies $\frac{\sqrt{C_2} n_b^{\beta_2 / 2}}{\delta_b} \sim n_b^{(2+\beta_2)/4}$, it indeed increases faster than $n_b^{1-\beta_1}$ if $4 \beta_1 + \beta_2 > 2$. Thus, \eqref{eq:delay_eps_eqn} is satisfied for large $b$.

\item Case 2: $4 \beta_1 + \beta_2 \leq 2$. We want the two terms to be simultaneously dominated by $n_b^{1-\beta_1}$. Let
\begin{align}
\label{eq:delay_choose_eps_delta_2}
    &\eps_b = 2 C_1 n_b^{1-\beta_1} \nonumber\\
    &\delta_b = \frac{2 f C_1}{I} n_b^{-\beta_1} \iff \frac{f \eps_b}{n_b I} = \delta_b.
\end{align}
The rest of the proof is similar, and in the same way we get \eqref{eq:delay_true_llr}.
To verify the first term in \eqref{eq:delay_eps_eqn} is dominated by $n_b^{1-\beta_1}$ with $\delta_b$ defined in \eqref{eq:delay_choose_eps_delta_2}, note that it satisfies $\frac{\sqrt{C_2} n_b^{\beta_2 / 2}}{\delta_b} \sim n_b^{\beta_1 + \beta_2/2}$, which increases no faster than $\eps_b$ if $4 \beta_1 + \beta_2 \leq 2$. When equality is achieved, redefine $\eps_b = \brc{C_1 + \frac{I \sqrt{C_2}}{2 f C_1}} n_b^{1-\beta_1}$ and one still gets the inequality in \eqref{eq:delay_eps_eqn}.
\end{itemize}

Also, in either of the cases, we get
\begin{equation*}
    \Prob{1}{\sum_{i=1}^{n_b} Z_i < b + \eps_b} \leq \Prob{1}{\sum_{i=1}^{n_b} Z_i < (1-\delta_b^2) n_b I} \leq \delta_b^2.
\end{equation*}
which follows by \cite[Appendix~B]{lai1998}, assuming that \eqref{eq:lai_lower} is true for $Z_i$'s. Thus,
\begin{align*}
    &\sup_{\nu \geq 1} \Prob{\nu}{\sum_{i=t}^{t+n_b-1} \widehat{Z}^{t+n_b-1,t}_i < b}\\
    &\leq \Prob{1}{\sum_{i=1}^{n_b} Z_i \leq b + \eps_b} + \Prob{1}{\frac{1}{n_b} \sum_{i=1}^{n_b} \left(Z_i - \widehat{Z}^{n_b,1}_i\right) \geq \frac{\eps_b}{n_b}}\\
    &\leq 2 \delta_b^2. \qedhere
\end{align*}
\end{proof}


\begin{theorem}
\label{thm:opt}
Suppose that $b_\alpha$ is chosen as in \eqref{eq:nonparam_b_alpha}, with a window size $m_\alpha$ large enough to satisfy \eqref{eq:malpha}, and suppose that \eqref{eq:lai_upper}, \eqref{eq:lai_lower} hold for the true log-likelihood ratio. Then $\widehat{\tau}(b_\alpha)$ with $\widehat{\tau}$ defined in \eqref{def:loo_cusum} solves the problem in \eqref{prob_def} asymptotically as $\alpha \to 0$. The worst case delay is
\begin{align*}
    \inf_{\tau \in \mathcal{C}_\alpha} \WADD{\tau} \sim \WADD{\widehat{\tau} \left(b_\alpha\right)} \sim \frac{\abs{\log \alpha}}{\KL{p_1}{p_0}}
\end{align*}
as $\alpha \to 0$.
\end{theorem}
\begin{proof}
The asymptotic lower bound on the delay follows from \eqref{eq:lai_lower} by using \cite[Thm.~1]{lai1998}. The asymptotic optimality of $\widehat{\tau}(b_\alpha)$ follows from Lemma~\ref{lem:fa} and Lemma~\ref{lem:delay}.
\end{proof}

\section{Numerical Results}
\label{sec:num_res}

In Fig. \ref{fig:loo_perf}, we study the performance of the proposed LOO-CuSum test defined in \eqref{def:loo_cusum} through Monte Carlo (MC) simulations when the pre-change distribution is ${\cal N}(0,1)$. The LOO-KDE (defined in \eqref{def:kde}) is used to estimate the density. The actual post-change distribution is ${\cal N}(0.5,1)$, but the LOO-CuSum test has no knowledge of it. The performance of the LOO-CuSum test is compared with that of the following tests:
\begin{enumerate}
    \item the CuSum test (in \eqref{def:cusum}), which assumes full knowledge of the post-change distribution;
    \item the WL-GLR-CuSum test, which assumes that the post-change distribution belongs to $\{{\cal N}(\theta,1)\}_{\theta > 0}$.
\end{enumerate}
Different window sizes are considered, among which window sizes of 100 and 200 are sufficiently large to cover the full range of delay.
It is seen that the expected delay of the LOO-CuSum test is close to that of the WL-GLR test for all window sizes considered. The results also validate the first-order asymptotic optimality of the LOO-CuSum test for large enough window size (Thm~\ref{thm:opt}).

\begin{figure}[tbp]
\centerline{\includegraphics[width=.8\textwidth]{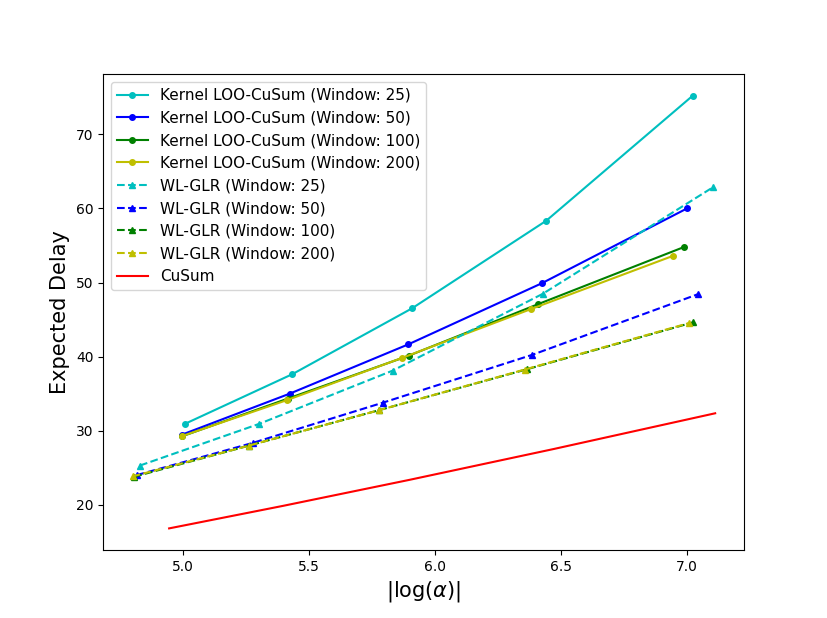}}
\vspace{-3mm}\caption{Comparison of operating characteristics of LOO-CuSum (solid lines) with the CuSum test (in red) and the WL-GLR-CuSum test (dotted lines) in detecting a shift in mean of a Gaussian. The pre- and post-change distributions are ${\cal N}(0,1)$ and ${\cal N}(0.5,1)$. The change-point $\nu = 1$. The kernel width parameter $h=(\min\{n, m\}-1)^{-1/5}$ where $m$ is the window size.
}
\label{fig:loo_perf}
\end{figure}

\section{Conclusion}
\label{sec:concl}

We studied a window-limited LOO-CuSum test for QCD that does not assume any knowledge of the post-change distribution, and does not require post-change training samples. We established the first-order asymptotic optimality of the test, and validated our analysis through numerical results.

\bibliographystyle{IEEEtran}
\bibliography{ref}

\begin{thebibliography}{10}
\providecommand{\url}[1]{#1}
\csname url@samestyle\endcsname
\providecommand{\newblock}{\relax}
\providecommand{\bibinfo}[2]{#2}
\providecommand{\BIBentrySTDinterwordspacing}{\spaceskip=0pt\relax}
\providecommand{\BIBentryALTinterwordstretchfactor}{4}
\providecommand{\BIBentryALTinterwordspacing}{\spaceskip=\fontdimen2\font plus
\BIBentryALTinterwordstretchfactor\fontdimen3\font minus
  \fontdimen4\font\relax}
\providecommand{\BIBforeignlanguage}[2]{{%
\expandafter\ifx\csname l@#1\endcsname\relax
\typeout{** WARNING: IEEEtran.bst: No hyphenation pattern has been}%
\typeout{** loaded for the language `#1'. Using the pattern for}%
\typeout{** the default language instead.}%
\else
\language=\csname l@#1\endcsname
\fi
#2}}
\providecommand{\BIBdecl}{\relax}
\BIBdecl

\bibitem{vvv_qcd_overview}
V.~V. Veeravalli and T.~Banerjee, ``Quickest change detection,'' in
  \emph{Academic press library in signal processing: Array and statistical
  signal processing}.\hskip 1em plus 0.5em minus 0.4em\relax Cambridge, MA:
  Academic Press, 2013.

\bibitem{xie_vvv_qcd_overview}
L.~Xie, S.~Zou, Y.~Xie, and V.~V. Veeravalli, ``Sequential (quickest) change
  detection: Classical results and new directions,'' \emph{IEEE Journal on
  Selected Areas in Information Theory}, vol.~2, no.~2, pp. 494--514, 2021.

\bibitem{lorden1971}
G.~Lorden, ``Procedures for reacting to a change in distribution,'' \emph{The
  Annals of Mathematical Statistics}, vol.~42, no.~6, pp. 1897--1908, Dec.
  1971.

\bibitem{lai1998}
T.~L. Lai, ``Information bounds and quick detection of parameter changes in
  stochastic systems,'' \emph{IEEE Transactions on Information Theory},
  vol.~44, no.~7, pp. 2917--2929, November 1998.

\bibitem{non_stat_glr_2021}
Y.~Liang and V.~V. Veeravalli, ``Quickest detection of composite and
  non-stationary changes with application to pandemic monitoring,'' in
  \emph{ICASSP 2022 - 2022 IEEE International Conference on Acoustics, Speech
  and Signal Processing (ICASSP)}, 2022, pp. 5807--5811.

\bibitem{huber1965}
P.~J. Huber, ``A robust version of the probability ratio test,'' \emph{The
  Annals of Mathematical Statistics}, vol.~36, no.~6, pp. 1753--1758, Dec.
  1965.

\bibitem{moulin-veeravalli-2018}
P.~Moulin and V.~V. Veeravalli, \emph{Statistical Inference for Engineers and
  Data Scientists}.\hskip 1em plus 0.5em minus 0.4em\relax Cambridge, UK:
  Cambridge University Press, 2018.

\bibitem{Molloy2017}
T.~L. {Molloy} and J.~J. {Ford}, ``Misspecified and asymptotically minimax
  robust quickest change detection,'' \emph{IEEE Transactions on Signal
  Processing}, vol.~65, no.~21, pp. 5730--5742, 2017.

\bibitem{Unnikrishnan2011}
J.~{Unnikrishnan}, V.~V. {Veeravalli}, and S.~P. {Meyn}, ``Minimax robust
  quickest change detection,'' \emph{IEEE Transactions on Information Theory},
  vol.~57, no.~3, pp. 1604--1614, 2011.

\bibitem{xie_mstat}
S.~Li, Y.~Xie, H.~Dai, and L.~Song, ``M-statistic for kernel change-point
  detection,'' in \emph{Advances in Neural Information Processing Systems},
  C.~Cortes, N.~Lawrence, D.~Lee, M.~Sugiyama, and R.~Garnett, Eds.,
  vol.~28.\hskip 1em plus 0.5em minus 0.4em\relax Curran Associates, Inc.,
  2015, pp. 3366--3374.

\bibitem{kernelcusum}
T.~Flynn and S.~Yoo, ``Change detection with the kernel cumulative sum
  algorithm,'' in \emph{2019 IEEE 58th Conference on Decision and Control
  (CDC)}, 2019, pp. 6092--6099.

\bibitem{one_svm}
F.~Desobry, M.~Davy, and C.~Doncarli, ``An online kernel change detection
  algorithm,'' \emph{IEEE Transactions on Signal Processing}, vol.~53, no.~8,
  pp. 2961--2974, 2005.

\bibitem{chu2018sequential}
L.~Chu and H.~Chen, ``Sequential change-point detection for high-dimensional
  and non-euclidean data,'' \emph{IEEE Transactions on Signal Processing},
  vol.~70, pp. 4498--4511, 2022.

\bibitem{nn_2019}
\BIBentryALTinterwordspacing
H.~Chen, ``Sequential change-point detection based on nearest neighbors,''
  \emph{The Annals of Statistics}, vol.~47, no.~3, pp. 1381--1407, 2019.
  [Online]. Available: \url{https://doi.org/10.1214/18-AOS1718}
\BIBentrySTDinterwordspacing

\bibitem{yilmaz_2017}
Y.~Yilmaz, ``Online nonparametric anomaly detection based on geometric entropy
  minimization,'' in \emph{2017 IEEE International Symposium on Information
  Theory (ISIT)}, 2017, pp. 3010--3014.

\bibitem{kurt_2020}
M.~N. Kurt, Y.~Yılmaz, and X.~Wang, ``Real-time nonparametric anomaly
  detection in high-dimensional settings,'' \emph{IEEE Transactions on Pattern
  Analysis and Machine Intelligence}, vol.~43, no.~7, pp. 2463--2479, 2021.

\bibitem{sugiyama}
Y.~Kawahara and M.~Sugiyama, ``Sequential change-point detection based on
  direct density-ratio estimation,'' \emph{Statistical Analysis and Data
  Mining: The ASA Data Science Journal}, vol.~5, no.~2, pp. 114--127, 2012.

\bibitem{moustakides2019}
G.~V. Moustakides and K.~Basioti, ``Training neural networks for
  likelihood/density ratio estimation,'' \emph{arXiv preprint
  arXiv:1911.00405}, Nov. 2019.

\bibitem{binning}
T.~Lau, W.~P. Tay, and V.~Veeravalli, ``A binning approach to quickest change
  detection with unknown post-change distribution,'' \emph{IEEE Transactions on
  Signal Processing}, vol.~67, no.~3, pp. 609--621, Nov. 2018.

\bibitem{de-with-kl-loss}
P.~Hall, ``On {K}ullback-{L}eibler loss and density estimation,'' \emph{The
  Annals of Statistics}, vol.~15, no.~4, pp. 1491--1519, 1987.

\bibitem{mult-denst-est}
D.~W. Scott, \emph{Multivariate Density Estimation: Theory, Practice, and
  Visualization}, 2nd~ed.\hskip 1em plus 0.5em minus 0.4em\relax Hoboken, NJ:
  John Wiley \& Sons, Inc., 2015.

\bibitem{page1954}
E.~S. Page, ``Continuous inspection schemes,'' \emph{Biometrika}, vol.~41, no.
  1/2, pp. 100--115, Jun. 1954.

\bibitem{moustakides1986}
G.~V. Moustakides, ``Optimal stopping times for detecting changes in
  distributions,'' \emph{Annals of Statistics}, vol.~14, no.~4, pp. 1379--1387,
  Dec. 1986.

\bibitem{tartakovsky_qcd2020}
A.~G. Tartakovsky, \emph{Sequential Change Detection and Hypothesis Testing:
  General Non-i.i.d. Stochastic Models and Asymptotically Optimal Rules}, ser.
  Monographs on Statistics and Applied Probability 165.\hskip 1em plus 0.5em
  minus 0.4em\relax Boca Raton, London, New York: Chapman \& Hall/CRC Press,
  Taylor \& Francis Group, 2020.

\end{thebibliography}


\end{document}